\documentclass[lettersize,journal]{IEEEtran}
\usepackage{amsmath,amsfonts}
\usepackage{algorithmic}
\usepackage{algorithm}
\usepackage{array}
\usepackage[caption=false]{subfig}
\usepackage{textcomp}
\usepackage{stfloats}
\usepackage{url}
\usepackage{verbatim}
\usepackage{graphicx}
\usepackage{cite}
\usepackage{multirow}
\usepackage{booktabs} 
\usepackage{amssymb}
\usepackage{amsmath}
\usepackage{mathtools}
\usepackage{amsthm}
\usepackage{graphicx}

\theoremstyle{plain}
\newtheorem{theorem}{Theorem}[section]

\newtheorem{lemma}[theorem]{Lemma}

\theoremstyle{definition}
\newtheorem{definition}[theorem]{Definition}

\theoremstyle{remark}

\begin{document}

\title{Quantifying Privacy Leakage in Split Inference via Fisher-Approximated Shannon Information Analysis}


\author{Ruijun Deng, Zhihui Lu,~\IEEEmembership{Member, IEEE}, and Qiang Duan,~\IEEEmembership{Senior Member, IEEE}, Shijing Hu

\thanks{Ruijun Deng and Shijing Hu are with the College of Computer Science and Artificial Intelligence, Fudan University, Shanghai, China, and with Engineering Research Center of Cyber Security Auditing and Monitoring, Ministry of Education, Shanghai, China (e-mail: rjdeng18@.fudan.edu.cn, sjhu24@m.fudan.edu.cn).}
\thanks{Zhihui Lu is with the College of Computer Science and Artificial Intelligence, Fudan University, Shanghai, China, and also with Shanghai Blockchain Engineering Research Center, Shanghai, China (e-mail: lzh@fudan.edu.cn).}
\thanks{Qiang Duan is with Information Sciences \& Technology Department, Pennsylvania State University, Abington, PA, USA (e-mail: qduan@psu.edu).}
\thanks{Corresponding authors: Zhihui Lu.}
}

\markboth{Journal of \LaTeX\ Class Files,~Vol.~14, No.~8, August~2021}%
{Shell \MakeLowercase{\textit{et al.}}: A Sample Article Using IEEEtran.cls for IEEE Journals}

\IEEEpubid{0000--0000/00\$00.00~\copyright~2021 IEEE}

\maketitle

\begin{abstract}
Split inference (SI) partitions deep neural networks into distributed sub-models, enabling collaborative learning without directly sharing raw data. However, SI remains vulnerable to Data Reconstruction Attacks (DRAs), where adversaries exploit exposed smashed data to recover private inputs. Despite substantial progress in attack–defense methodologies, the fundamental quantification of privacy risks is still underdeveloped. This paper establishes an information-theoretic framework for privacy leakage in SI, defining leakage as the adversary’s certainty and deriving both average-case and worst-case error lower bounds. We further introduce Fisher-approximated Shannon information (FSInfo), a new privacy metric based on Fisher Information (FI) that enables operational and tractable computation of privacy leakage. Building on this metric, we develop FSInfoGuard, a defense mechanism that achieves a strong privacy–utility tradeoff. Our empirical study shows that FSInfo is an effective privacy metric across datasets, models, and defense strengths, providing accurate privacy estimates that support the design of defense methods outperforming existing approaches in both privacy protection and utility preservation. The code is available at https://github.com/SASA-cloud/FSInfo.
\end{abstract}


\section{INTRODUCTION}\label{sec:1}


 Machine learning as a service (MLaaS) allows users to access machine learning (ML) services by simply uploading the raw data to be inferred. To enable privacy-preserving MLaaS, model partition, a mainstream research line of distributed machine learning, has been widely employed to segment deep neural networks (DNNs) into multiple sub-models and deploy them on different hosts (e.g., mobile phones and edge servers). 
 This learning paradigm has been adopted in numerous model-distributed co-training and co-inference; for example, split/collaborative inference (SI) \cite{zhang2021autodidactic,ICLR21-two-party-label}, split learning (SL) \cite{deng2023hsfl,ijcai2024p596, oh2022locfedmix}, vertical federated learning (VFL) \cite{chen2023practical}, pipeline parallelism \cite{yuan2023pipeedge,ye2024galaxy}, and their combinations \cite{duan2022combined}. W.l.o.g, we focus on a two-party \textbf{split inference system} (see Fig.~\ref{fig:DRA}) where the inputs are first fed into the bottom model on the client, and then the intermediate smashed data are transmitted to the top model on the server for further prediction, thus hopefully preserving raw data privacy. 
\begin{figure}[!t]
  \centering
  \includegraphics[width=0.85\linewidth]{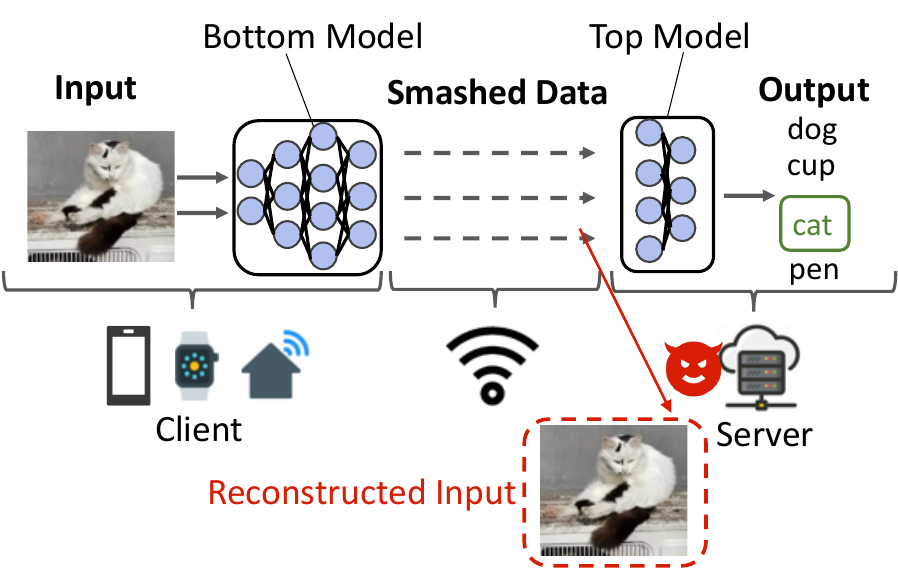}
  \caption{The information flow in a split inference system and the data reconstruction attack.}
    \label{fig:DRA}
\end{figure}

Despite transmitting only smashed data rather than raw inputs, the SI paradigm remains susceptible to privacy inference attacks. Adversaries can execute Membership Inference Attacks (MIAs) to identify training data participation \cite{sec21-evaluation-ml,wu2024quantifying}, Attribute Inference Attacks (AIAs) to deduce sensitive attributes (e.g., demographic information) \cite{liu2022ml}, and Data Reconstruction Attacks (DRAs) to recover original user data \cite{blackbox,yang2022measuring,pasquini2021unleashing}. Our research focuses on DRAs due to their heightened severity: successful execution compromises entire datasets rather than isolated attributes, potentially causing substantial damage.
\textbf{Privacy leakage quantification} (or privacy quantification for short) is an independent concept that transcends the dualistic framework of privacy attack and defense. This critical but less investigated aspect of privacy issues on ML is to develop an appropriate measure (i.e., a privacy metric) of private information leakage from the exposed data (e.g., model parameters, gradients, and smashed data) about the raw input data. In this paper, privacy quantification means assessing the private information exposure from smashed data under DRAs in the SI system of Fig.~\ref{fig:DRA}. Privacy quantification serves a dual purpose: to assess the privacy risks of ML systems for oversight and compliance, and to guide the design of effective defenses by informing architecture and method selection.
\IEEEpubidadjcol

Privacy quantification methods in ML primarily fall into two categories: attack-based and theory-based. Attack-based methods \cite{wu2024quantifying,hu2023quantifying,deng2023hsfl} assess privacy leakage via inference attacks, using success rates as empirical metrics. Notably, advances in attack strategies can be directly applied for quantification. \textit{While practical, these methods depend heavily on adversary strength and lack theoretical rigor, often underestimating privacy risks.}  
In contrast, theory-based methods use formal theoretical frameworks such as metric differential privacy (metric-DP) \cite{singh2023posthoc}, information theory (mutual information; MI) \cite{arevalo2024task,noorbakhsh2024inf2guard}, and hypothesis testing (Fisher information; FI) \cite{maeng2023bounding} to infer privacy leakage. \textit{However, these approaches are either difficult to compute accurately or overly sensitive to outliers, risking overestimation.}

{\bf Contributions.} To address the above two limitations, in this paper, first, we theoretically define privacy leakage as the adversary’s certainty based on information theory and provide guarantees against DRA adversaries. Second, we employ FI to identify the strongest adversary and derive a lower bound on the defined privacy leakage, yielding a practical, training-free metric: Fisher-approximated Shannon Information (FSInfo). 
Then, with FSInfo, we deveolop a defense method FSInfoGuard for defending against DRA adversaries without significant performance degradation. Finally, we empirically validate FSInfo as both an effective assessment tool and a design guide for SI systems and defense strategies. Experiments on image, tabular, and text data show that FSInfo correlates well with DRA performance, and defenses guided by FSInfo outperform existing methods in privacy-utility tradeoffs. We also analyze three key design factors—data distribution, model width/depth, and overfitting—to provide practical insights for system development.

In summary, the contributions of our paper are as follows:
\begin{itemize}
    \item We formalize privacy leakage against DRAs in SI systems within an information-theoretic framework and derive both average-case and worst-case lower bounds on the adversary’s reconstruction error.
    \item We introduce a new privacy metric, FSInfo, which leverages FI to operationally compute the defined leakage. Guided by FSInfo, we design FSInfoGuard, a defense that injects Gaussian noise into the smashed data to mitigate DRAs with minimal performance loss.
    \item We conduct extensive experiments to validate the effectiveness of FSInfo and FSInfoGuard and analyze three key factors—dataset characteristics, model size, and overfitting—that influence privacy leakage.
\end{itemize}

\section{RELATED WORK}\label{sec:2}
The privacy metrics \cite{CSUR-18-metrics-survey}, although having been studied decades ago in the areas of smart meter \cite{kalogridis2010privacy} and database \cite{agrawal2001design}, have not received much attention in the field of machine learning until recently \cite{deng2024invmetrics}. 

Current privacy metrics primarily emerge as byproducts of privacy attacks, demonstrating that attacks developed for breaching privacy can also be applied to quantify it \cite{wu2024quantifying,hu2023quantifying,deng2023hsfl,AsiaCCS20-1dCNN,pasquini2021unleashing,liu2022ml}. For instance, Liu et al. \cite{liu2022ml} developed the ML-Doctor framework, employing four distinct privacy inference attacks to assess private data and model parameter vulnerabilities. \textit{However, these empirical approaches, while intuitive, lack theoretical foundations and exhibit adversary-dependent limitations, potentially yielding overly optimistic privacy assessments with weaker adversaries.}



Theoretical approaches offer deeper privacy analysis through advanced mathematical tools. Pan et al. \cite{pan2022exploring} quantify privacy leakage via activated neuron counts, modeling it as a linear equation system's complexity. However, this method primarily applies to fully connected networks (FCNs) with ReLU activations, limiting its effectiveness for nonlinear functions. While differential privacy (DP) remains a prominent DNN privacy framework, traditional local DP proves unsuitable for split inference \cite{ICLR21-two-party-label,singh2023posthoc,sec24-defending-FL-ita}. Singh et al. \cite{singh2023posthoc} address this limitation by generalizing DP definition to geodesic on the data manifold. \textit{Yet DP's inherent indistinguishability property, while effective against membership inference, demonstrates reduced efficiency and limited efficacy against reconstruction attacks, as empirically shown in \cite{sec24-defending-FL-ita}.}

Information-theoretic approaches, particularly Shannon mutual MI \cite{mireshghallah2020shredder,noorbakhsh2024inf2guard,arevalo2024task} and distance correlation \cite{deng2024invmetrics,sun2022label,sun2021defending}, quantify privacy leakage by measuring the information shared between $A$ and $B$. 
Researchers have applied these measures to various data pairs: input features and smashed data \cite{mireshghallah2020shredder,noorbakhsh2024inf2guard,deng2024invmetrics,sun2021defending}, training data and model parameters \cite{sec24-defending-FL-ita}, and labels and smashed data \cite{sun2022label}. However, the closed-form information quantities (e.g., Shannon MI) of arbitrarily high-dimensional random variables are still an open problem, prompting recent approaches to circumvent direct MI calculation through noise injection \cite{sec24-defending-FL-ita}. Therefore, in ML, sample-based information quantities estimations have been proposed \cite{cheng2020club}. For example, \cite{mireshghallah2020shredder} uses MI calculation tools based on K-nearest neighbor entropy estimation, and \cite{arevalo2024task,sec24-defending-FL-ita,nips24-club} use CLUB, an upper bound of MI and parameterized variational distributions to quantify privacy leakage. \textit{However, these learning-based estimators may fail to approximate reliably when the data dimension increases \cite{cheng2020club}. Also, the tightness of these bounds relies on the choice of variational distribution and the accuracy of trained models, which may lead to approximation inaccuracies.}

FI metrics \cite{hannun2021measuring,guo2022bounding,maeng2023bounding}, computable with a single backward pass, estimate DRA adversary error based on hypothesis testing theory. \textit{However, they are susceptible to outlier activations of smashed data, potentially overestimating privacy leakage.}

\section{PRELIMINARIES}\label{sec:3}
\subsection{Split Inference Pipeline}
W.l.o.g, we consider a two-party split inference paradigm between the client $D_1$ and the server $D_2$ as shown in Fig.~\ref{fig:DRA}. 
We focus on the privacy leakage from the smashed data of split inference about the raw input data. $\mathcal{U} = \{(x_i,y_i)_{i=1}^N \subseteq \mathbb{R}^{d_x} \times \mathbb{R}^{d_y}\}$ is a dataset corresponding to intelligent ML services. The two parties learn from $\mathcal{U}$ with a split DNN $f_\theta(\cdot) = (f_{\theta2}\circ f_{\theta1})(\cdot)=f_{\theta2}(f_{\theta1}(\cdot))$, where $D_1$ holds the bottom model $f_{\theta1}$ and $D_2$ the top model $f_{\theta2}$. The split point (SP) $p\in \{1,2,\dots,L\}$ is the last layer of the bottom model $f_{\theta1}$, where $L$ is the total number of layers of $f_\theta$.
The procedure starts with that $D_1$ takes as input the raw data $x$ and generates smashed data $z=f_{\theta 1}(x)$ as output.
Then, $P_2$ forwards $z$ to get the predictions: $y=f_{\theta 2}(z)$.

\subsection{Threat Model} \label{sec:3-DRA}

We focus on the common setting of DRAs, where the client $D_1$ with private data is benign and the server $D_2$ is untrusted.
\paragraph{Adversary's Goal.} 
In a data reconstruction attack (model inversion attack) \cite{yin2023ginver,blackbox}, the goal of the adversary (edge server) $\mathcal{A}$ is to reconstruct the value of raw data $x$: $\hat{x} = g_\phi(z)$,
where the $g_\phi$ is a attack function parameterized by $\phi$, and $\hat{x}$ is the reconstructed input.

\paragraph{Adversary's Knowledge \& Capability.} \label{sec:TM}
First, as many DRAs rely on different prior knowledge to conduct reconstruction, fewer assumptions about the adversary's knowledge are preferred for broadening FSInfo's applicability. Therefore, we do not leave tight requirements on the adversaries’ knowledge by allowing them to know and decide whether to use or not the plain-text smashed data $z$ and all side information. i.e., the auxiliary dataset $\mathcal{V}$ with the same distribution of $\mathcal{U}$, architecture and parameters of bottom model $f_{\theta1}$. This is compatible with most attacks \cite{blackbox,yin2023ginver,yang2022measuring} and the same as other privacy metrics \cite{maeng2023bounding, pasquini2021unleashing}. 
Then, for simplicity, we only consider the unbiased attack function $g_\phi$, but one can apply our metric to the biased scenarios by using score matching techniques in \cite{maeng2023bounding}.

\section{Privacy Leakage Quantification}\label{sec:4}


\subsection{Privacy and Privacy Leakage}
The semantic meaning of privacy (information) depends on the diverse attributes of users and application contexts, e.g., culture and law. 
In contrast, privacy leakage, no matter what the privacy is (e.g., age or facial features), mainly depends on the ML task and threat model, i.e., the specific adversaries and their capabilities. For example, against the DRA adversaries in image classification tasks, privacy leakage is the value of the input data $x$. Against MIA adversaries, privacy leakage is the membership of a single data $x$, i.e., the knowledge of whether or not $x$ is used in the training data of the model \cite{singh2023posthoc}.


\subsection{Privacy Leakage Formulation}\label{PL-formulate}


In a SI system, we treat the raw input and the corresponding output of bottom model $f_{\theta1}$ as two multi-dimensional random variables, $X$ and $Z$, characterized by distributions $P(X)$ and $P(Z|  X,\theta1)$. After observing the $Z$, the adversary $\mathcal{A}$ out put the estimation $\hat{X}$ of $X$. This process forms a Markov chain $X \rightarrow Z \rightarrow \hat{X}$. 

In information theory, {\itshape entropy} $H(X)$ is a measure of the uncertainty of a random variable. Formally, 
we consider the privacy leakage as the certainty or confidence of $\hat{X}$ given $X$ as defined below:
\begin{definition} \label{def:1}
({\bf Privacy Leakage)}. Privacy Leakage about raw input data $X$ from smashed data $Z$ of model $f_{\theta1}$ against adversary $\mathcal{A}$ is the negative conditional entropy of the estimation $\hat{X}$ given raw input data $X$:
\begin{equation}
\mathcal{L}^{\mathcal{A}}_{f_{\theta1},X}= -H(\hat{X}\vert X), 
\end{equation}
where $\hat{X}$ is the reconstructed input of $\mathcal{A}$ and a larger negative conditional entropy indicates more confident estimation $\hat{X}$ given $X$ and less error of the DRA adversary, thus larger privacy leakage. 
\end{definition}
We now give the relation between the $\mathcal{L}^{\mathcal{A}}_{f_{\theta1},X}$ and the average lower bound of the adversary's error:
\begin{theorem} \label{theorem:REB-Tan}
({\bf Average-Case Reconstruction Error Bound}). For two random variables raw input $X$ and reconstructed input $\hat{X}$ with the same dimension $d_x$, we have:
\begin{equation}
    \frac{\mathbb{E}_{p(X,\hat{X})}[\Vert X-\hat{X} \Vert^2]}{d_x} \geq \frac{e^{\frac{2}{d_x}H(\hat{X}\vert X)}}{2\pi e},
\end{equation}
where the lower bound of the adversary's error is smaller under a higher privacy leakage. The proof is provided in the Appendix. 
\end{theorem}

Definition~\ref{def:1} captures the adversary's uncertainty of $\hat{X}$ given $X$, serving as a naive criterion for measuring the recovery hardness. However, first, this definition remains agnostic to adversarial power, i.e., how accurate estimation the adversary can achieve. For instance, neural network (NN)-based adversaries typically achieve superior reconstruction accuracy compared to maximum likelihood estimation (MLE)-based approaches as its $p(\hat{X})$ is closer to the $p(X)$ \cite{blackbox}.
Second, the corresponding Theorem~\ref{theorem:REB-Tan} only gives average error over all possible adversaries. But privacy is not an average-case metric \cite{aerni2024evaluations}, and it is necessary for the worst-case privacy leakage analysis to avoid damages resulting from the strongest adversary.


{\bf Strongest Adversary (SA).} To address this limitation, we extend the definition of privacy leakage by incorporating adversaries' capabilities. We define the ``strongest'' adversary as the one whose estimate $\hat{X}^s$ achieves the highest accuracy among all possible estimators. For example, if the estimate is unbiased, the SA corresponds to the one with the minimum covariance. Note that ``strongest'' refers to estimation accuracy, not to the adversary's prior knowledge or its ability to influence the SI process. Importantly, this worst-case analysis is fundamentally different from the data-independent worst-case analysis used in DP.


Following this, we introduce the definition of privacy leakage against the strongest adversary:
\begin{definition} \label{def:2}
({\bf Privacy Leakage against the SA)}. Privacy Leakage about raw input data $X$ from smashed data $Z$ of model $f_{\theta1}$ against the strongest adversary $\mathcal{A}^s$ is the negative conditional entropy of estimation $\hat{X}^s$ given raw input data $X$:
\begin{equation}
    \mathcal{L}^{\mathcal{A}^s}_{f_{\theta1},X}= -H(\hat{X}^s\vert X),
\end{equation}
where $\hat{X}^s$ is the reconstructed input of $\mathcal{A}^s$.
\end{definition}
We then derive the relation between the privacy leakage definition in Definition~\ref{def:2} and the worst-case lower bound of the adversary's reconstruction error. We adopt the statistical minimax estimation setting \cite{scarlett2019introductory},
to denote estimation loss as:
\begin{equation}\label{minimax-risk}
    \ell(X,\hat{X}) = \Phi(\rho(X,\hat{X})),
\end{equation}
where $\rho(X,\hat{X})$ is a metric and $\Phi(\cdot)$ is an increasing function maps from $\mathbb{R}^+$ to $\mathbb{R}^+$. Equation~\ref{minimax-risk} is a general formulation shared among many loss functions. For example, the MSE loss ($\ell(X,\hat{X})=\frac{1}{d_x}\Vert X-\hat{X}\Vert_2^2$) used by \cite{maeng2023bounding} as the privacy definition
clearly takes this form \cite{scarlett2019introductory}.
The minimax risk of the adversary's estimation is
\begin{align}
    \mathcal{M}(\mathcal{X},\ell)&=\inf_{\hat{X}\in\mathcal{\hat{X}}}\sup_{X\in\mathcal{X}}\mathbb{E}_{X}[\ell(X,\hat{X})]\\
    &=\sup_{X\in\mathcal{X}}\mathbb{E}_{X}[\ell(X,\hat{X}^s)],
\end{align}
where $\hat{X}$ are lie in subset $\mathcal{X}$. $\mathcal{M}(\mathcal{X},\ell)$ can be interpreted as the largest estimation error achieved by the SA or the least privacy leakage captured by the SA in this paper.

Then we can give the lower bound of the adversary's minimax risk with the variation of Fano inequality \cite{scarlett2019introductory} as:
\begin{theorem} \label{theorem:REB-minimax}
({\bf Worst-Case Reconstruction Error Bound}). Let random variables $X$, $Z$ and $\hat{X}$ form a Markov chain $X \rightarrow Z \rightarrow \hat{X}$ and $X$ lie in the $\mathcal{X}$, such that
\begin{equation}
    \rho(x_i,x_j)\geq \epsilon, \forall{i,j} \in \{1,\dots,|\mathcal{X}|\}, i\neq j.
\end{equation}
Then, we have:
\begin{equation}
    \mathcal{M}(\mathcal{X},\ell)\geq\Phi(\frac{\epsilon}{2})(1-\frac{H(\hat{X}^s)-H(\hat{X}^s\vert X)+log2}{log(|\mathcal{X}|)}),
\end{equation}
where $X$ is uniform over $\mathcal{X}$.
\end{theorem}
Specifically, $\ell$ represents a general loss function adaptable to specific application scenarios, and increased privacy leakage $\mathcal{L}^{\mathcal{A}^s}_{f_{\theta1},X}= -H(\hat{X}^s\vert X)$ directly corresponds to a lower minimax risk and less adversarial estimation error.
The proof for Theorem~\ref{theorem:REB-minimax} is provided in the Appendix.

\subsection{Computing Privacy Leakage}\label{sec:5}
Definition~\ref{def:2} gives a privacy leakage definition as adversaries' certainty. However, it can not be used directly for practitioners to evaluate the security level of a split inference framework, as the calculation of conditioned entropy or the MI of arbitrary high-dimensional random variables is a long-standing open problem \cite{noorbakhsh2024inf2guard,arevalo2024task,xiao2023pac}. 
Therefore, as shown in Fig.~\ref{fig:FMI}, we translate Fisher information into a Shannon information quantity to estimate the certainty of the ``strongest'' adversary, i.e. $\mathcal{L}^{\mathcal{A}^s}_{f_{\theta1},X}$. Then we use the Gaussian distribution to get maximum conditional entropy, offering an operational privacy metric, FSInfo.

\begin{figure}[!t]
  \centering
  \includegraphics[width=\linewidth]{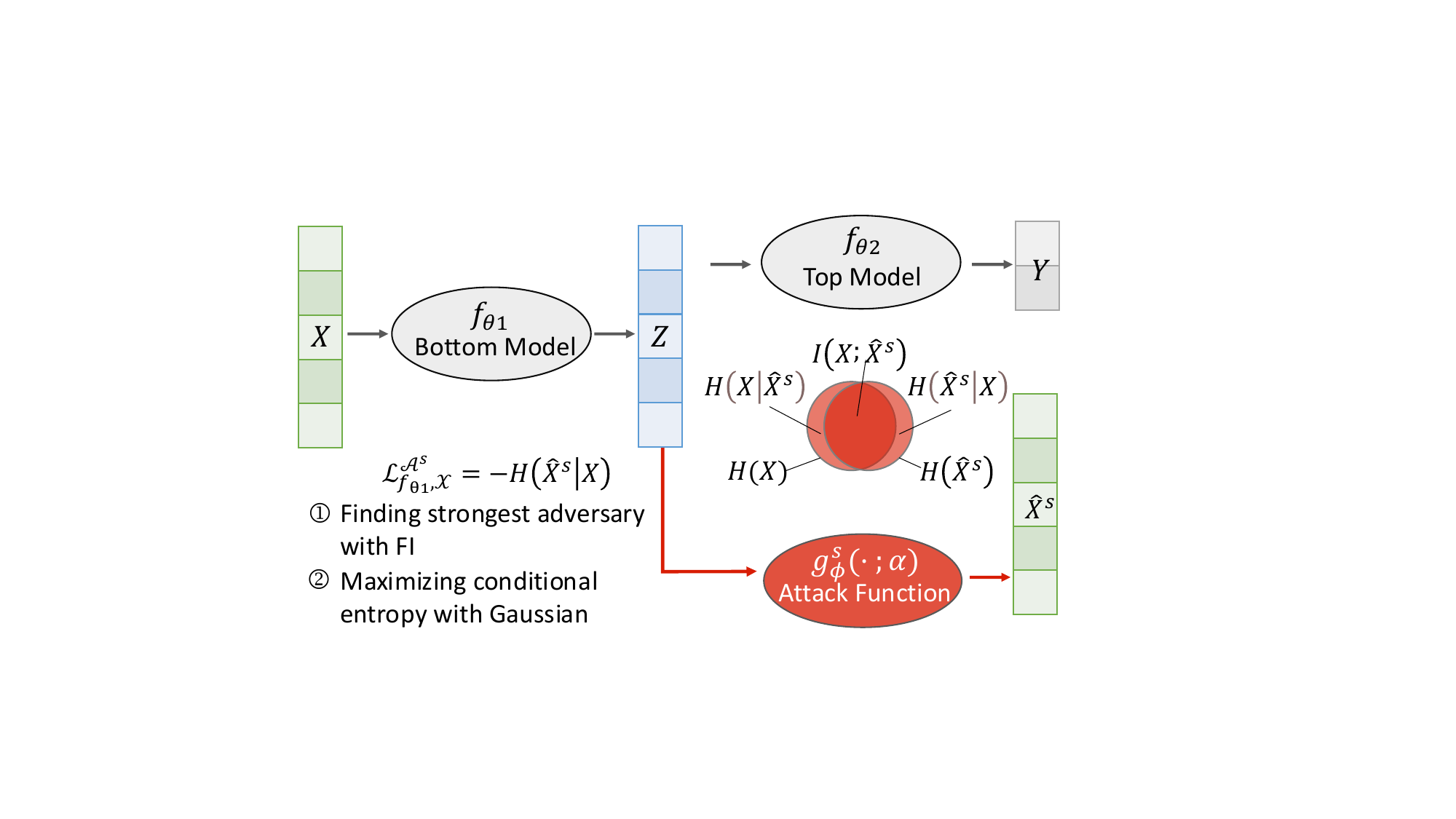}
  \caption{The proposed FSInfo privacy metric.}
    \label{fig:FMI}
\end{figure}



{\bf Fisher information} is a fundamental concept in statistics hypothesis testing \cite{martens2020new}, which is related to the asymptotic variability of an estimator, where higher FI means lower estimation error:
\begin{definition}
({\bf Fisher Information Matrix) \cite{cover1999elements}}. For two random variables $x \sim p(x)$ and $z\sim p(z\vert x)$, the element of Fisher Information Matrix (FIM) or the Fisher Information (FI) $F_{z\vert x}$ of $z$ w.r.t. the $x$ is defined as:
\begin{align}
    F_{z\vert x}(i,j) &= \mathbb{E}_z[(\frac{\partial}{\partial x_i}\log p(z \vert x))(\frac{\partial}{\partial x_j}\log p(z \vert x))],
\end{align}
where $\log p(z \vert x)$ is the log conditional probability density function (p.d.f.) of $z$ given $x$. 
\end{definition}
Note that, FI is defined based on the distribution; however, in the SI, the neural network $f_{\theta1}$ is a deterministic function, and $p(z\vert x)=p(f_{\theta1}(x) \vert x) = 1$. Therefore, we follow \cite{maeng2023bounding,martens2020new} to fit FI into DL models by combining $f_{\theta1}$ with a random noise: $\tilde{f}_{\theta1}(\cdot)=r(f_{\theta1}(\cdot))$, where $r(z)=z+\delta, \delta \sim \mathcal{N}(0, \sigma^2)$. To this end, we can calculate FI of raw input $x$ contained by the smashed data $z$ as:
\begin{align}
    F_{z\vert x} =& \frac{1}{\sigma^2}J_{f_{\theta1}}(x)^\top J_{f_{\theta1}}(x), \label{eq:fisher-comp}
\end{align}
where $J_{f_{\theta1}}(x)$ is the Jacobian of $z$ with respect to the raw input $x$.


The significance of FI is shown in the following theorem:
\begin{theorem}\label{theorem-CRB}
({\bf Cram\'{e}r-Rao Bound)} \cite{cover1999elements}. The mean-squared error of any unbiased estimator $T(z)$ of the parameter $x$ is lower bounded by the reciprocal of the Fisher information:
\begin{align}
    Cov(T) \ge F_{z|x}^{-1},
\end{align}
where $\geq$ is used in the Löwner partial order.
\end{theorem}


Theorem~\ref{theorem-CRB} establishes the optimality of the SA. Considering the DRA adversary as an estimator $T(z)$, the strongest adversary should have the minimal covariance $F_{Z|X}^{-1}$ and a mean of $\mathbb{E}[X]$ (according to the unbiased assumptions in Section~\ref{sec:TM}). The $H(\hat{X}|X)$, no matter what types of distribution the $p(\hat{X}|X)$ and $p(\hat{X})$ are, is smaller than the entropy of a Gaussian distribution with the same covariance and mean. Because, with the same mean and covariance, the Gaussian distribution has the largest entropy \cite{sec24-defending-FL-ita}. Therefore, the possibly least captured leakage of the SA
is: $-H(\mathcal{N}(x,F_{Z|X}^{-1}))=-\frac{1}{2}\log(\frac{(2\pi e)^{d_x}}{det(F_{Z|X})})$. Consequently, we summarize the above discussion into the following Theorem~\ref{theorem-LP-of-SA}.

\begin{theorem}\label{theorem-LP-of-SA}
{\bf(Privacy Leakage Lower Bound of SA)}. A Lower Bound of the privacy leakage against the strongest adversary is:
\begin{align}
        \mathcal{L}^{\mathcal{A}^s}_{f_{\theta1},X}=-H(\hat{X}^s \vert X)\geq -\frac{1}{2}\log(\frac{(2\pi e)^{d_x}}{det(F_{Z|X})}), \label{eq:PL-SA-Cal}
\end{align}
where the 
$det(F_{Z|X})$ is the determinant of $F_{Z|X}$. 
\end{theorem}




Now, with the raw input $X$, we can directly calculate the lower bound of privacy leakage against the strongest adversary $\mathcal{A}^s$ in Definition~\ref{def:2} through Equation~\ref{eq:PL-SA-Cal},
where all the terms can be easily implemented.


{\bf Acceleration.} The calculation of $ -\frac{1}{2}\log(\frac{(2\pi e)^{d_x}}{det(F_{Z|X})})$ involves computing the Jacobian of $z$ w.r.t. $x$, which is time-consuming when the $d_x$ gets large (e.g., for a $32\times32\times3$ image in CIFAR10 \cite{krizhevsky2009learning}, the number of rows of $F_{Z|X}$ is 3072). We speed up this process by first approximating $F_{Z|X}$ with its diagonal version, which is widely used in transfer learning \cite{fan2021discriminative,xue2021toward}:
\begin{align}
-\frac{1}{2}\log(\frac{(2\pi e)^{d_x}}{det(F_{Z|X})})
        =& - \frac{1}{2}(d_x \log(2\pi e) - \sum_{i}^{d_x} \log(\Lambda_i)),
\end{align}
where the $\Lambda_i$ is the diagonal element of $F_{Z|X}$. Second, we use subsampling techniques such as random sampling and average pooling to reduce input dimensions $d_x$.

{\bf Privacy Leakage per Dimension.} Although $\mathcal{L}^{\mathcal{A}^s}$ can now outline the privacy leakage risk of an SI system, its value exhibits considerable variation w.r.t. the dimensionality of the input data, e.g., it tends to increase with the dimension of the input data.
Therefore, we adopt the average leakage across all dimensions. FSInfo is a computable privacy metric, defined as an unbiased estimate of this average leakage using data pairs $\{x, y\}\in \mathcal{U}$:
\begin{definition} \label{def:FSInfo}
    ({\bf FSInfo}). In an SI system, the privacy metric FSInfo about $X$ in the smashed data $Z=f_{\theta1}(X)$ is:
    \begin{align}
        FSInfo 
        =& - \frac{1}{2d_x}\mathbb{E}_{x\sim p(X)}[d_x \log(2\pi e) - \sum_{i}^{d_x} \log(\lambda_i)], \label{eq:Ifisher}  
    \end{align} 
    where the $\lambda_i$ is the diagonal element of $F_{z|x}$.
\end{definition}
Definition~\ref{def:FSInfo} gives the essence link in DRA between the Fisher information and Shannon information, where the maximal Shannon information gain can be achieved by adversary $\mathcal{A}^s$ whose estimation error is lower bounded by the Fisher information.

\subsection{Guiding Calculation of Defensive Noise}
FSInfo not only quantifies privacy leakage in SI systems but also provides a risk signal that can guide the design of defense mechanisms. Based on this idea, we develop FSInfoGuard, a defense against DRAs for SI systems.

Privacy protection in split inference differs fundamentally from federated learning. In SI, smashed data preserves a one-to-one mapping to the raw input rather than aggregating information, meaning that batch size—an important factor in FL privacy—has limited influence on SI privacy leakage \cite{geiping2020inverting}. Consequently, obfuscation-based defense mechanisms, particularly those utilizing regularization or closed-form noise injection techniques, have become predominant in split inference systems \cite{singh2024simba}. Regularization methods optimize the client-side model with heuristic regularization terms in the loss function to guide it to output smashed data that is hard to invert. Closed-form defenses directly add noise with an analytical noise scale on the smashed data to cater for the desired security level (e.g., metric-DP \cite{singh2023posthoc}).


These obfuscation-based defenses typically rely on a privacy-related objective or feedback signal to perturb the smashed data. \textbf{Nopeek} \cite{ICDM-20-nopeek} and \textbf{Shredder} \cite{mireshghallah2020shredder} are representative regularization-based approaches. Nopeek adopts adversarial representation learning to encourage the bottom model to produce smashed data containing less information about the raw input, using the loss: $\alpha DLoss(x,z) + CCE$, where $DLoss$ is a privacy metric estimates the correlation between $x$ and $z$ and $CCE$ is the cross-entropy loss. Shredder instead trains a noise layer to generate Gaussian perturbations that explicitly reduce the estimated mutual information (a privacy metric) between $x$ and $z$. \textbf{inv\_dFIL\_def} is a closed-form defense that adds Gaussian noise with scale $\sigma=\sqrt{\frac{\text{Tr}(J^{\top}_{f_{\theta 1}} J_{f_{\theta 1}})}{d \times \text{dFIL}}}$, where $f_{\theta1}$ is the bottom model, $d$ is the input dimension, and dFIL is the target privacy metric. This noise guarantees the desired dFIL privacy level.

FSInfo can serve as a standalone privacy metric for designing both regularization-based and closed-form defenses. In principle, one could minimize FSInfo directly by adding it to the training loss; however, this empirical regularization strategy offers no provable security guarantees. Motivated by this limitation, we introduce \textbf{FSInfoGuard}, a closed-form noise-based defense that injects Gaussian noise calibrated to meet a target FSInfo level. From Definition~\ref{def:FSInfo} and Equation~\ref{eq:fisher-comp}, achieving a specific FSInfo requires adding Gaussian noise with zero mean and scale $\sigma = \frac{\det(J_{f_{\theta 1}}^TJ_{f_{\theta 1}})^\frac{1}{2d}}{e^{FSInfo}(2\pi e)^\frac{1}{2}}$ to the smashed data. FSInfoGuard applies this noise during both training and inference: during training, the model learns from perturbed smashed data, and during inference, the injected noise suppresses privacy leakage and mitigates DRAs.




\section{Experiments}\label{sec:6}
In this section, we first introduce the experimental setups. Then, we evaluate FSInfo across various models, datasets, and defenses to validate its effectiveness and highlight the prevalence of privacy leakage. We then show how FSInfo can guide the design of privacy-preserving architectures and methods by comparing the privacy-utility tradeoff of FSInfoGuard with other defenses. Finally, we analyze three factors that influence privacy leakage: data characteristics, model size, and overfitting.


\subsection{Setups}
\subsubsection{Datasets and Models}

We use four image datasets (CelebA \cite{liu2015faceattributes}, CIFAR-10 \cite{krizhevsky2009learning}, TinyImageNet \cite{le2015tiny}, and MNIST \cite{lecun1998gradient}), two tabular datasets (Purchase100 \cite{sec21-evaluation-ml} and Home Credit \cite{home-credit-kaggle}), and one text dataset (GLUE/SST2 \cite{wang2018glue}) for our experiments. We use the train set for model training, and the test set for FSInfo computing. We use VGG-5 \cite{simonyan2014very}, VGG-9 \cite{simonyan2014very}, ResNet-18 \cite{he2016deep}, ResNet-34 \cite{he2016deep} for image data, FCNs for tabular data, and DistilBert \cite{sanh2019distilbert} for text data. The number of critical layers and the architecture of these models implemented are presented in the Appendix. The split layer is identified by the index of that critical layer, starting from 1. 
We only present the results for the image dataset for lack of space, and the additional results can be found in the Appendix.


\subsubsection{Attacks and Defenses}\label{compared-methods}

Following \cite{noorbakhsh2024inf2guard}, we adopt an NN-based DRA approach \cite{blackbox}. We adopt three commonly used metrics to evaluate the attack performance: visual invertibility, Structural Similarity Index Measure (SSIM), and Mean Squared Error (MSE). 
\underline{Visual invertibility} measures the similarity of the raw input and the reconstructed one by human perception \cite{sun2024privacy}. We visualize the reconstructed images and observe their clarity and similarity to the raw input. This metric offers semantic judgment for privacy information observed by the adversary. A higher visual quality means a larger privacy leakage.
\underline{SSIM} is used to assess the perceived quality of generated images, which aligns closely with human visual perception. A higher SSIM indicates greater privacy leakage.
\underline{MSE} is also an objective criterion, which measures the error or distance between two images. A large MSE indicates less privacy leakage.
Note that the attack efficacy only establishes a lower bound on the actual privacy leakage, rather than the ground truth---at least, what the attacker has already obtained has been leaked. 


We use three defense methods in the experiments, including two regularization defenses (Nopeek \cite{ICDM-20-nopeek} and Shredder \cite{mireshghallah2020shredder}) and one closed-form defense (inv\_dFIL\_def \cite{maeng2023bounding}). First, we adjust the defense strength of Nopeek and inv\_dFIL\_def to observe the corresponding FSInfo. Second, we compare the privacy-utility tradeoff between FSInfoGuard and these three defenses.

\subsection{Evaluating FSInfo among Different Datasets}

\begin{figure*}[t]
  \centering
  \subfloat[VGG-5, VGG-9, ResNet-18 on CIFAR-10.\label{subfig:mse-cifar10}]{
    \includegraphics[width=0.45\linewidth]{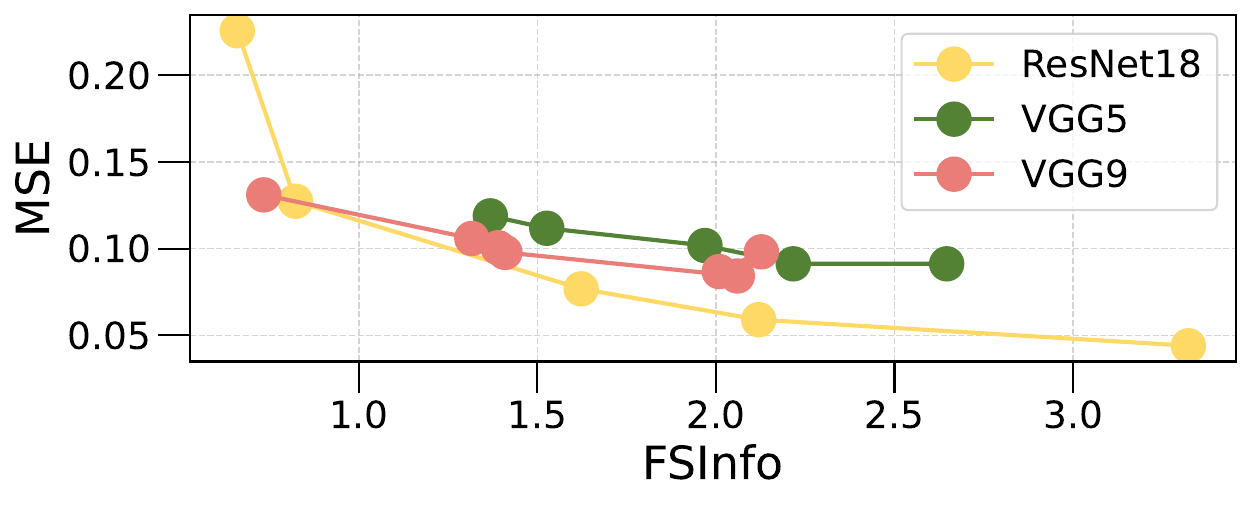}
  }
  \subfloat[VGG-5 and VGG-9 on MNIST.\label{subfig:mse-mnist}]{
    \includegraphics[width=0.45\linewidth]{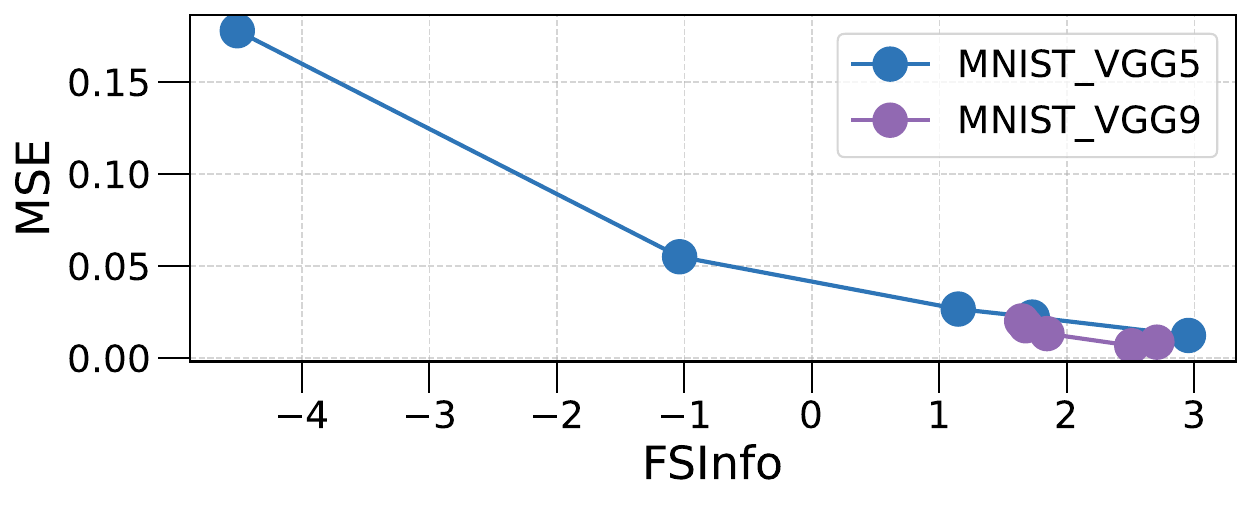}
  }
  \vspace{0.01cm}
  \subfloat[DistilBert on GLUE/SST2.\label{subfig:mse-distilbert}]{
    \includegraphics[width=0.45\linewidth]{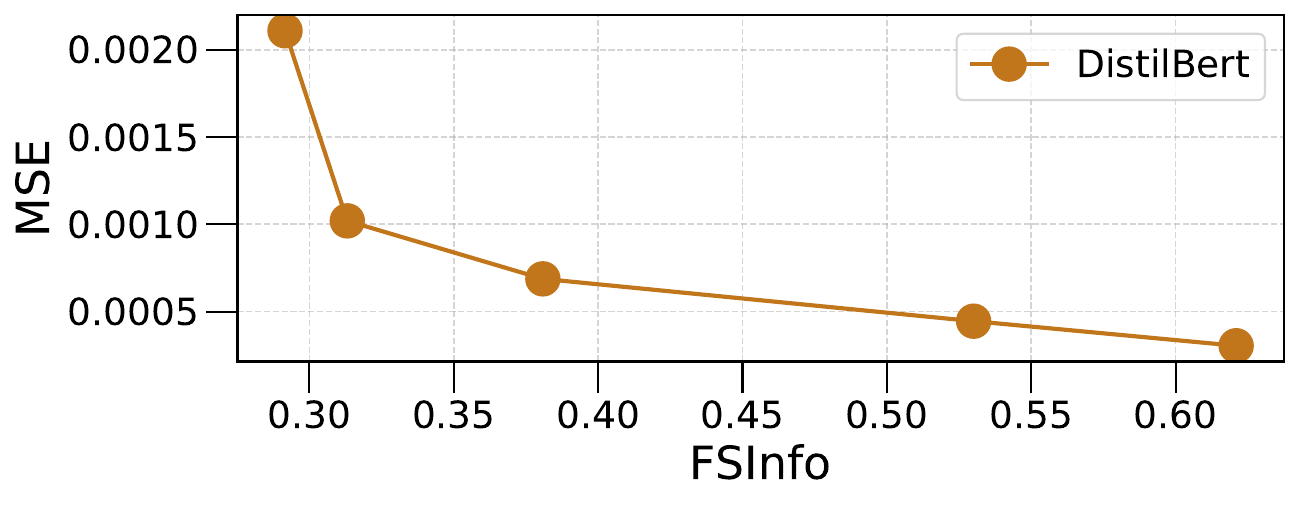}
  }
  \subfloat[FCNs on Purchase100 and Home Credit.\label{subfig:mse-tabular}]{
    \includegraphics[width=0.45\linewidth]{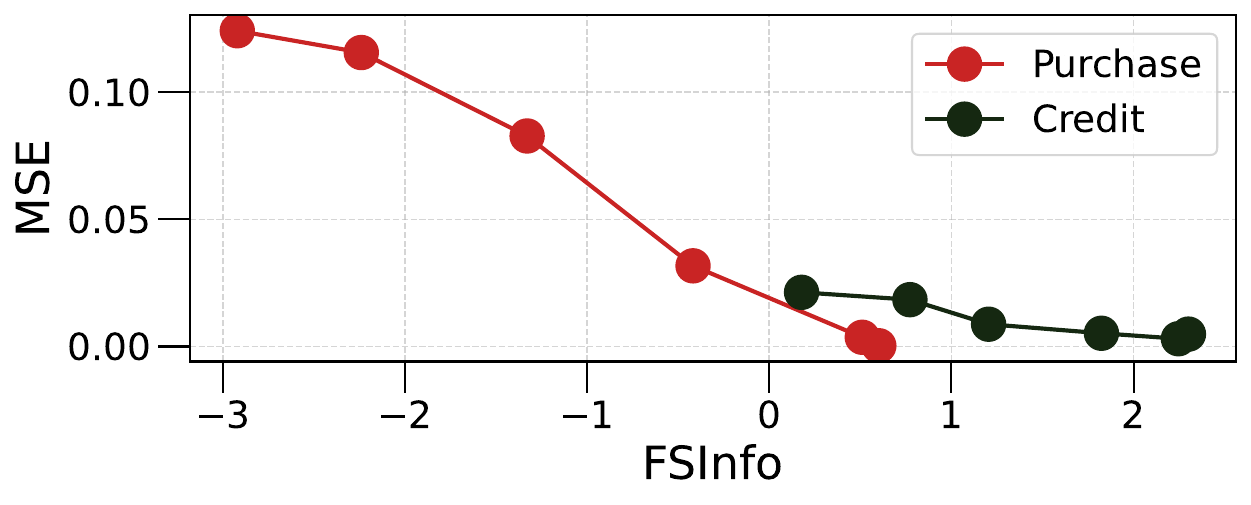}
  }
  \caption{MSE vs. FSInfo across different datasets and models.}
  \label{fig:exp1-mse-diff-data}
\end{figure*}

In this section, we compute FSInfo on multiple modalities to examine whether it effectively reflects the degree of privacy leakage and aligns with the DRA adversary’s MSE. For each dataset, we vary the model’s split point to obtain smashed data with different exposure levels, enabling a clearer comparison between FSInfo and actual attack performance. For image (CIFAR-10, MNIST) and tabular (Purchase, Credit) datasets, we apply the inverse-network DRA and compute FSInfo of the smashed data with respect to the raw inputs. For text data, the discrete inputs are tokenized and passed through an embedding layer before entering the transformer. In this case, following Maeng et al. \cite{maeng2023bounding}, we adopt an MLE-based DRA \cite{blackbox} and compute FSInfo with respect to the token embeddings.


Fig.~\ref{fig:exp1-mse-diff-data} presents the layer-averaged FSInfo and reconstruction MSE. \textbf{Overall, FSInfo shows the expected negative correlation with MSE across all datasets}: higher FSInfo indicates greater privacy leakage, enabling the DRA adversary to reconstruct inputs more accurately. This aligns with the information-theoretic interpretation of FSInfo—larger FSInfo implies more exposed information and reduced uncertainty for the adversary. However, \textbf{the correlation is less clear for VGG-9}, where the final data point shows an upward trend (Fig.~\ref{fig:exp1-mse-diff-data}(a)). This may indicate that the adversary did not fully exploit the leaked information or that its capability is limited relative to the leakage level. The last two points correspond to the 4th and 2nd layers of VGG-9; the 4th layer produces higher-dimensional smashed data, which likely allows the adversary to capture leaked information more effectively than the 2nd layer.


These evaluations collectively demonstrate that FSInfo remains effective across multiple data types (image, tabular, text), model architectures (FCNs, CNNs, transformers), and application scenarios. Overall, FSInfo is highly consistent with actual attack performance and reliably tracks changes in privacy leakage under different system configurations, supporting its effectiveness as a privacy risk indicator.

\subsection{Evaluating FSInfo among Defenses}
\begin{figure*}[t]
  \centering
  \subfloat[inv\_dFIL\_def]{
    \label{subfig:dFIL-fsinfo}
    \includegraphics[width=0.48\linewidth]{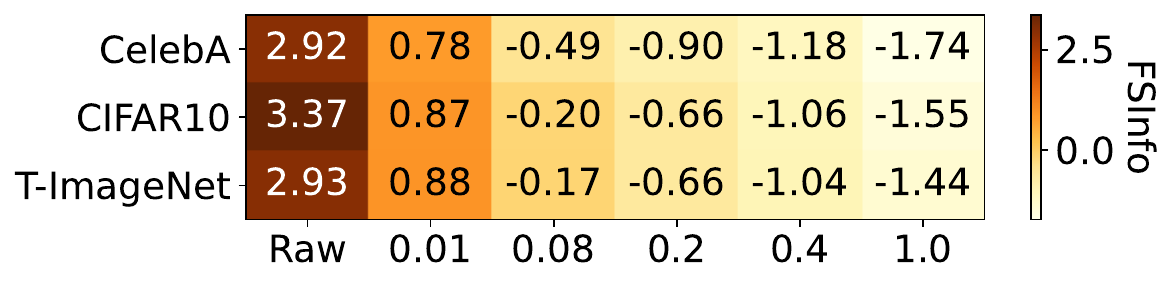}
  }
  \subfloat[Nopeek]{
    \label{subfig:nopeek-fsinfo}
    \includegraphics[width=0.48\linewidth]{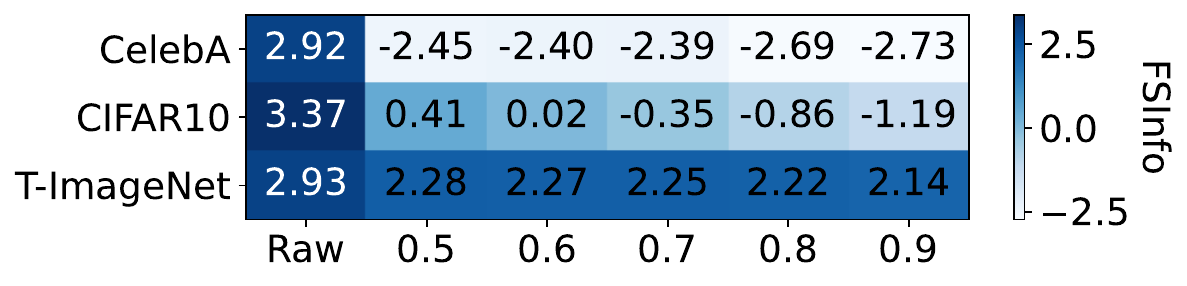}
  }
  \caption{FSInfo under different defenses and strengths.}
  \label{fig:exp1-defense-eval}
\end{figure*}

As a privacy metric, FSInfo can be used to analyze and monitor privacy risks in DNN-based systems. In this section, we demonstrate its practical relevance by showing that practitioners can also use FSInfo to assess the split inference system employed by defense mechanisms. We examine how privacy leakage changes under different defense strengths using ResNet-18 on CIFAR-10, CelebA, and TinyImageNet (T-ImageNet), with the split point set at the first convolutional layer. We apply FSInfo to SI systems equipped with one regularization-based defense (Nopeek \cite{ICDM-20-nopeek}) and one closed-form defense (inv\_dFIL\_def \cite{maeng2023bounding}). Defense strength is controlled via hyperparameters: for Nopeek, a larger regularization weight $\alpha$ implies stronger protection, while for inv\_dFIL\_def, a higher expected MSE lower bound $\frac{1}{dFIL}$ corresponds to stronger defense.





Fig.~\ref{fig:exp1-defense-eval} reports FSInfo values under different defense strengths. The horizontal axis denotes the defense hyperparameters, and the vertical axis corresponds to the datasets. \textbf{Overall, stronger defenses lead to lower privacy leakage, demonstrating the consistency of FSInfo across datasets and defense mechanisms.} An exception occurs on CelebA when the Nopeek hyperparameter increases from 0.5 to 0.7, where privacy leakage slightly increases. This suggests that at low defense strengths, the client model’s optimization may still be dominated by the cross-entropy loss rather than the regularization term.

\subsection{Evaluating FSInfo by Visual Perception}
\begin{figure}[t]
  \centering
  \includegraphics[width=0.8\linewidth]{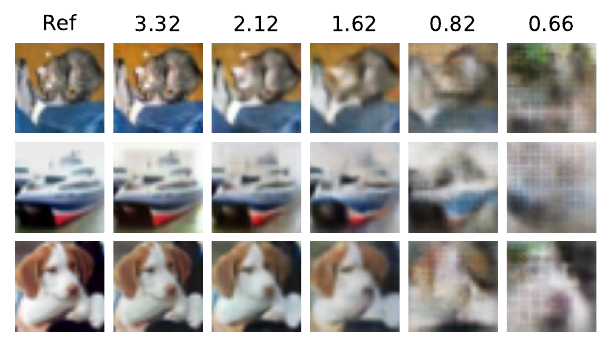}
  \caption{Visualization of raw input (left) and the reconstructed input of ResNet-18 on CIFAR-10 among different split points vs. their FSInfo.}
    \label{fig:exp1-visual-resnet18-cifar10}
\end{figure}
\begin{figure*}[t]
  \centering
  \subfloat[VGG-5, VGG-9, ResNet-18 on CIFAR-10.]{
    \includegraphics[width=0.45\linewidth]{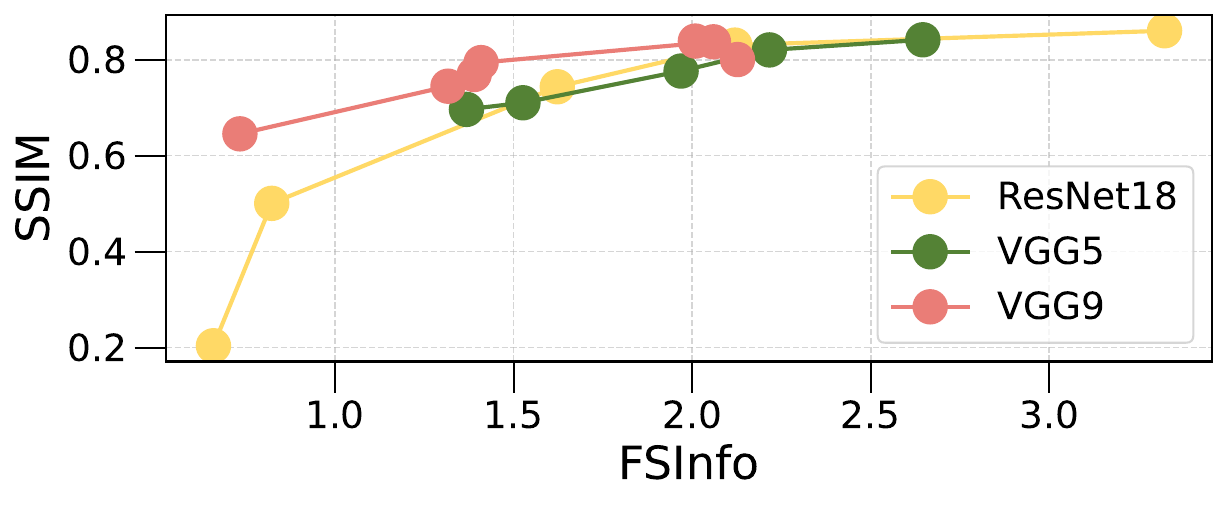}
  }
  \subfloat[VGG-5 and VGG-9 on MNIST.]{
    \includegraphics[width=0.45\linewidth]{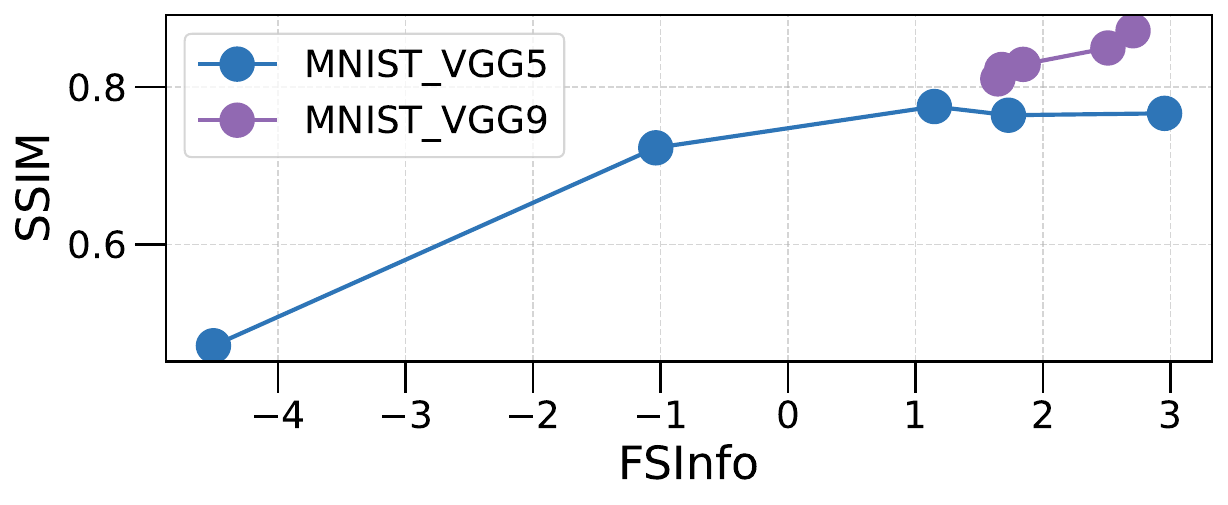}
  }
  \caption{SSIM vs. FSInfo on CIFAR-10 and MNIST.}
  \label{fig:exp1-ssim-diff-data}
\end{figure*}

We next examine whether FSInfo aligns with human perception of reconstructed images. Using CIFAR-10, we visualize reconstructions from different layers of ResNet-18 and compute the layer-averaged FSInfo. As shown in Fig.~\ref{fig:exp1-visual-resnet18-cifar10}, image clarity decreases as FSInfo decreases, indicating a strong correspondence between FSInfo and perceived visual quality. To further validate this, we plot SSIM against FSInfo in Fig.~\ref{fig:exp1-ssim-diff-data}. Overall, \textbf{SSIM positively correlates with FSInfo, suggesting that greater privacy leakage (higher FSInfo) yields reconstructions that are more discernible to humans}. The slight downward trend for VGG-9 in Fig.~\ref{fig:exp1-ssim-diff-data}(a) likely stems from the same cause noted earlier (Fig.~\ref{fig:exp1-mse-diff-data}(a))—the adversary does not fully exploit the leaked information. Moreover, in Fig.~\ref{fig:exp1-ssim-diff-data}(b), the SSIM trend for VGG-9 diverges from that of MSE (Fig.~\ref{fig:exp1-mse-diff-data}(b)), illustrating that empirical attack metrics can sometimes disagree with each other.

\subsection{Evaluating FSInfoGuard}
\begin{figure}[t]
  \centering
  \includegraphics[width=0.8\linewidth]{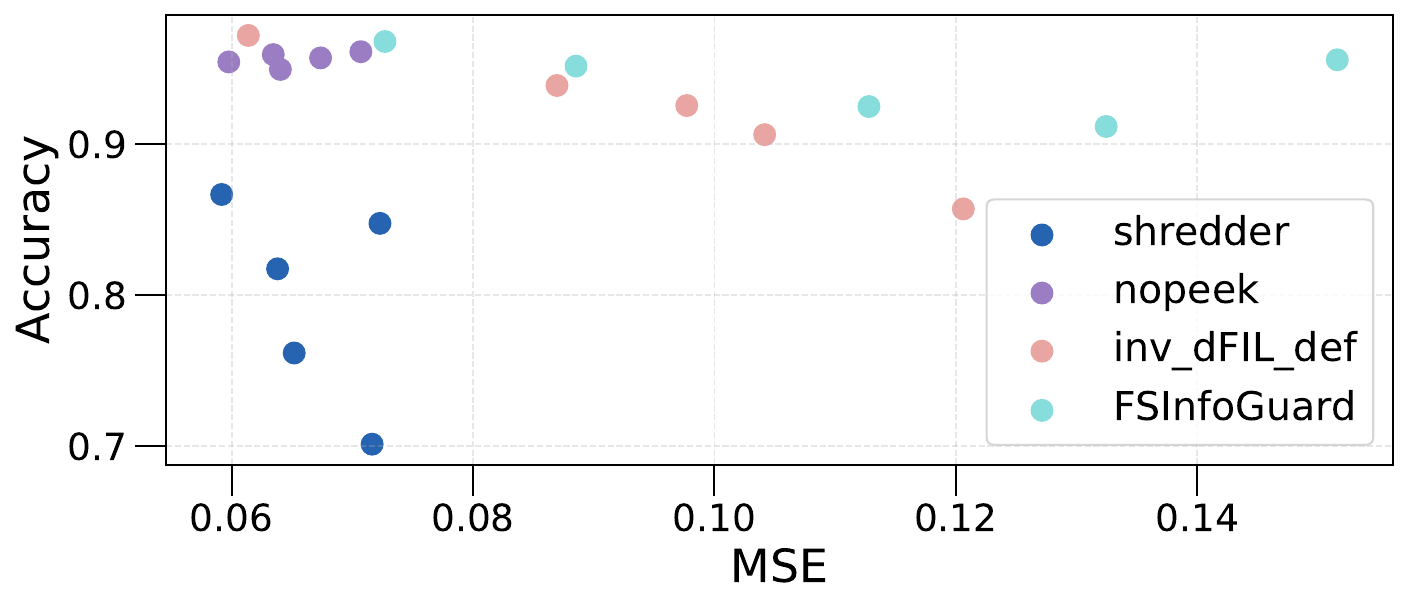}
  \caption{Task accuracy v.s. MSE on CelebA and ResNet-18.}
    \label{fig:pat-celebA-layer2}
\end{figure}


This section examines whether FSInfo can be used to directly optimize the privacy–utility trade-off (PUT) in practical defense design by comparing FSInfoGuard with existing defenses. In split inference, defensive methods obfuscate the smashed data before transmission to reduce privacy leakage. However, when these defenses rely on privacy metrics that inaccurately estimate leakage, they often introduce excessive noise, leading to significant performance degradation and a poor PUT.



We conduct experiments on CelebA using ResNet-18 and compare FSInfoGuard with Shredder, Nopeek, and inv\_dFIL\_def.
Shredder and inv\_dFIL\_def derive their Gaussian noise from their respective privacy metrics, whereas Nopeek only optimizes the bottom model without adding noise.
Defense strength is controlled by varying each method’s hyperparameters. As shown in Fig.~\ref{fig:pat-celebA-layer2}, where the $x$-axis indicates reconstruction MSE and the $y$-axis indicates task accuracy, \textbf{FSInfoGuard achieves the strongest privacy–utility trade-off}. This demonstrates that FSInfo provides more effective guidance for designing privacy-preserving mechanisms and accurately quantifies privacy leakage, enabling noise calibration that protects privacy without substantially degrading performance.

\subsection{The Impact of Intrinsic Data Characteristics}
\begin{figure}[t]
  \centering
  \includegraphics[width=0.8\linewidth]{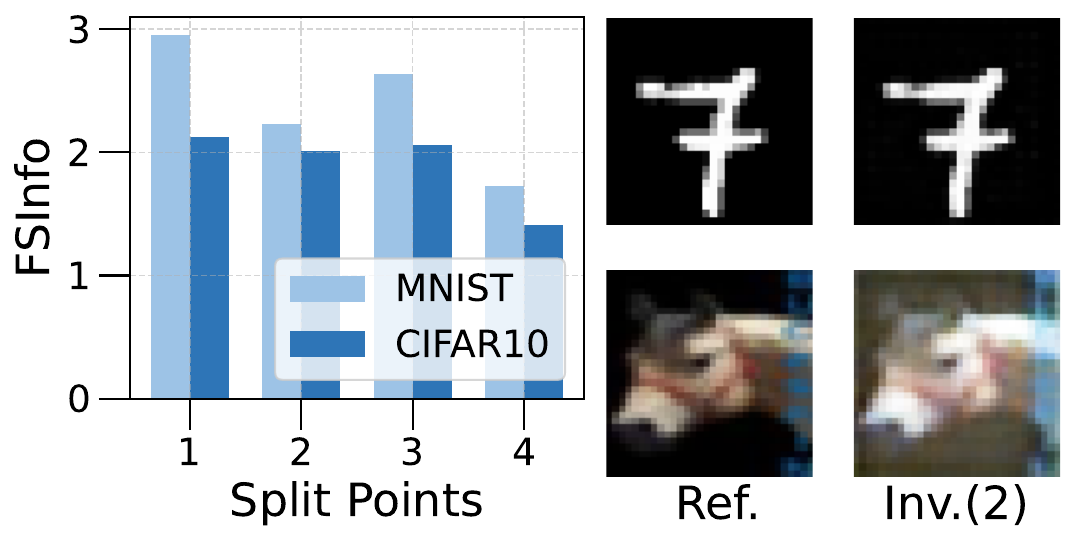}
  \caption{The impact of dataset complexity on privacy leakage. The left diagram is the FSInfo among different split layers of VGG-5. The right diagram is the raw inputs (Ref.) and the reconstructed inputs from split point 2 (Inv.(2)).}
  \label{fig:exp2-data}
\end{figure}
\begin{figure}[t]
  \centering
  \includegraphics[width=\linewidth]{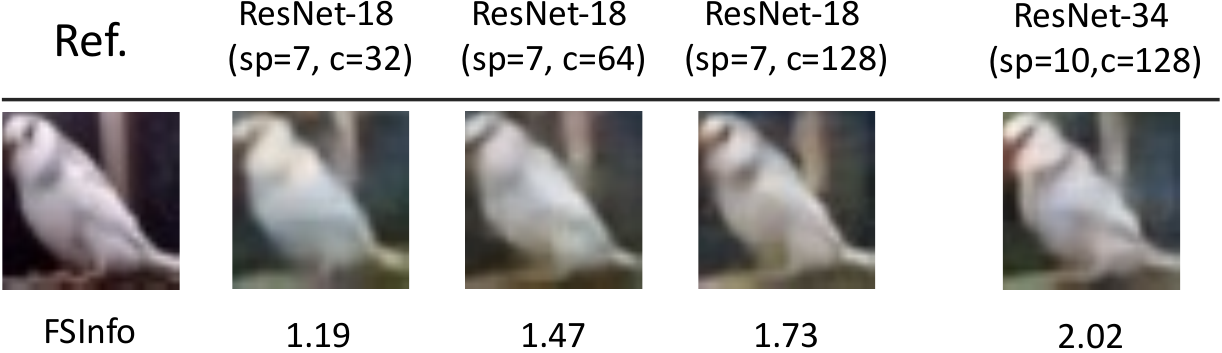}
  \caption{The impact of model size on privacy leakage on different models. Reconstructed inputs, raw input (Ref.), and FSInfo are displayed. The ResNet-18 architecture is displayed for three different base block channel sizes.}
  \label{fig:exp2-wandd1}
\end{figure}
\begin{figure}[t]
  \centering
  \includegraphics[width=0.9\linewidth]{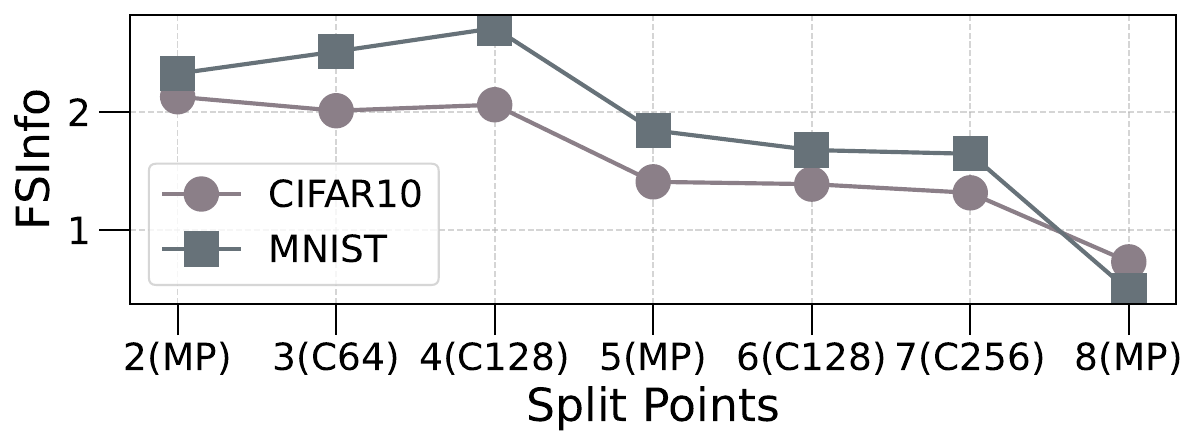}
  \caption{The impact of model size on privacy leakage on VGG-9 with different split points.}
  \label{fig:exp2-wandd2}
\end{figure}
\begin{figure}[t]
  \centering
  \includegraphics[width=0.9\linewidth]{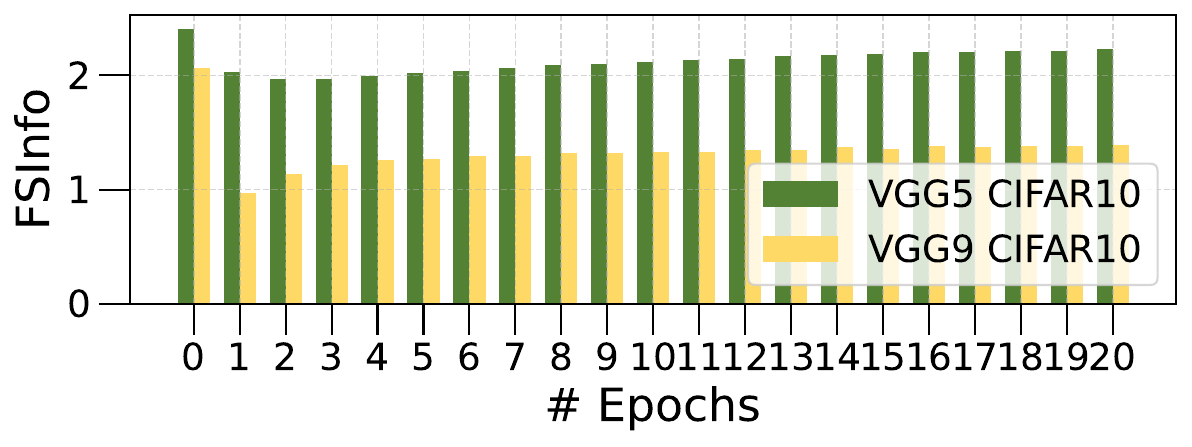}
  \caption{The impact of overfitting levels on privacy leakage. Epoch 0 means the model is just initialized and has not been trained.}
  \label{fig:exp2-overfitting}
\end{figure}

CIFAR-10 comprises diverse RGB images across 10 classes, whereas MNIST contains grayscale handwritten digits. Normalized to $[-1,1]$, CIFAR-10 exhibits richer pixel value diversity compared to MNIST's bimodal distribution peaking at -1 and 1 (see Appendix). To investigate data characteristics' impact on privacy leakage, we analyze FSInfo for both datasets across four split points (Fig.~\ref{fig:exp2-data}, left). Results demonstrate consistently lower privacy leakage for CIFAR-10, aligning with Liu et al.'s findings \cite{liu2022ml} that simpler images facilitate reconstruction. We attribute this discrepancy to dataset complexity and value space dimensions. Visual comparisons of raw and reconstructed images (Fig.~\ref{fig:exp2-data}, right) reveal that MNIST adversaries primarily handle binary color patterns, while CIFAR-10 reconstruction requires capturing complex color distributions, suggesting that richer information content increases reconstruction difficulty.

\subsection{The Impact of Model Width and Depth}

The model size (number of layers and the output size of a layer) influences neural network privacy characteristics. Fig.~\ref{fig:exp2-wandd1} demonstrates that both reconstruction quality and FSInfo increase with channel size in ResNet-18, as wider networks exhibit more pronounced neuron activation patterns. Interestingly, the deeper network (ResNet-34) may not always decrease the privacy leakage, because after being trained over 20 epochs, the ResNet-34 has a higher test accuracy than ResNet-18 on CIFAR-10, which may lead to a higher privacy leakage (see Section~\ref{sec:overfit}). Jiang et al. \cite{jiang2022fedsyl} suggest that additional layers hinder leakage through increased transformations, which is consistent with the results of Fig.~\ref{fig:exp2-wandd2} when SP is from 5 to 6. However, Fig.~\ref{fig:exp2-wandd2} reveals this intuition may be incomplete, as later split points (e.g., SP is from 3 to 4) can exhibit greater leakage with increased width. Notably, at SP=8, MNIST shows lower leakage than CIFAR-10 due to MaxPooling's differential impact on feature map dimensions (7x7→3x3 vs. 8x8→4x4). Overall, network width substantially amplifies privacy leakage, while the increased depth degrades privacy leakage more moderately. 


\subsection{The Impact of Overfitting}\label{sec:overfit}


Following established methodologies \cite{liu2022ml,sec21-evaluation-ml,Sec19-secret-share}, we refer to the training epochs as the overfitting level. Fig.~\ref{fig:exp2-overfitting} illustrates the relationship between epoch count and privacy leakage (FSInfo) for both VGG-5 and VGG-9, revealing a characteristic U-shaped pattern. Initially, untrained models exhibit high sensitivity to input stimuli, with a substantial gradient scale of smashed data w.r.t. the raw input responses leading to significant information leakage. As training progresses, this sensitivity gradually diminishes. However, prolonged training enables model parameters to memorize training data information, which subsequently manifests in smashed data during consequent testing \cite{achille2019information}, explaining the observed first-decline-then-rise FSInfo trajectory.




\section{Conclusion}
We correlate the privacy leakage with the DRA adversary's certainty and formulate it with conditional entropy, which can lower bound the adversary's reconstruction error in both average- and worst-case. We obtain an operationally computable privacy metric, FSInfo, by transforming Fisher information into Shannon quantity, and we design a defense method, FSInfoGuard, with a strong utility-privacy trade-off. 
Experimental results show that FSInfo reliably tracks actual privacy leakage across models, datasets, and defense strengths, exhibiting strong consistency with empirical attack performance and even human-perceived reconstruction quality. Also, FSInfo can be used to guide the design of defense mechanisms and to provide insights into key system components. Future work may extend the analysis of privacy leakage to other exposed information carriers (e.g., gradients in SL). Besides, the analysis of the privacy leakage determinants is conducted experimentally; future works may study it theoretically and derive more interpretable results.



{\appendices

\section*{More Related Work}\label{app:r&w}
\subsection{Data Reconstruction Attacks to SI}
Data Reconstruction Attacks (DRAs) employ two primary approaches: neural network (NN)-based methods, which learn mappings from observable outputs to raw inputs, and maximum likelihood estimation (MLE)-based methods, utilizing optimization techniques like stochastic gradient descent \cite{blackbox}. In federated learning, where only aggregated parameters and gradients are observable, MLE-based methods dominate, typically optimizing gradient matching objectives (Euclidean or cosine similarity) \cite{zhang2023survey,geiping2020inverting,jin2021cafe}.

SI scenarios favor NN-based methods due to the direct one-to-one correspondence between raw inputs and smashed data, providing ample training pairs for inverse neural networks to achieve a better attack performance than the MLE-based \cite{yin2023ginver,pasquini2021unleashing}. Recent methods in this line of strategy have focused on adjusting the training set, loss function, and decoder network architecture to prevent performance degradation in more practical scenarios. For example, the adversary does not have auxiliary datasets \cite{yin2023ginver} or face defense mechanisms \cite{yang2022measuring}. We aim to evaluate the system's robustness and provide guidelines for defense. In this case, we set loose conditions for the adversaries and allow them to use the basic but well-performing inverse network method \cite{blackbox} in our experiments.



\subsection{Information in Deep Neural Network}
In the information theory, {\itshape entropy} $H(X)$ is a measure of the uncertainty of a random variable $X$ and (Shannon) mutual information $I(X;Y)$ between random variables $X$ and $Y$ is the reduction in the uncertainty of $X$ due to the knowledge of $Y$, and vice versa \cite{cover1999elements}:
\begin{equation}
    I(X;Y) = H(X)-H(X\vert Y) = H(Y)-H(Y\vert X).
\end{equation}
Previous studies \cite{tishby2015deep,yu2020understanding} analyze DNNs by seeing the inputs, predictions, and outputs of one hidden layer (representations) as random variables $X$, $Y$, and $Z$. Model optimization can be done by minimizing $I(X;Z)$ and maximizing $I(Z;Y)$ under the information bottleneck (IB) tradeoff \cite{shwartz2017opening,yu2020understanding}.

\section*{Experimental Implementations}\label{app:model&datasets}
\subsection{Implementation}
Our experiments were performed on a server equipped with two NVIDIA GeForce RTX 4090 24GB GPUs, two Intel Xeon Silver 4310 CPUs, 256GB of memory, and Debian 6.1.119-1. We implemented the proposed metric with Python 3.10 and Pytorch 2.4.1. 

\subsection{Datasets}



\begin{table}[!h]
    \caption{Statistics of Datasets} 
    \label{tab:datasets}

    \centering
    \begin{tabular}{cccccc}
    \toprule
    Category & Dataset  & \# Train & \# Test & \# Feature & \# Class\\
    \midrule
    \multirow{2}{*}{Tabular} & Purchase100 & 158k & 39k & 600  & 100 \\
                             & Home Credit & 246k & 62k & 121 & 2 \\
    \midrule
    \multirow{4}{*}{Image}   & CIFAR-10 & 50k & 10k & 32x32x3 & 10 \\
                             & MNIST & 60k & 10k & 28x28x1 & 10\\
                             & CelebA & 163K & 20k & 178x218x3 & 2\\
                             & TinyImageNet & 100k & 10k & 64x64x3 & 4\\
    \midrule
    Text                     & GLUE/SST-2 & 67k & 1.8k & - & 2\\
    \bottomrule
    \end{tabular}
\end{table}

\begin{figure}[ht]
  \centering
  \includegraphics[width=0.6\linewidth]{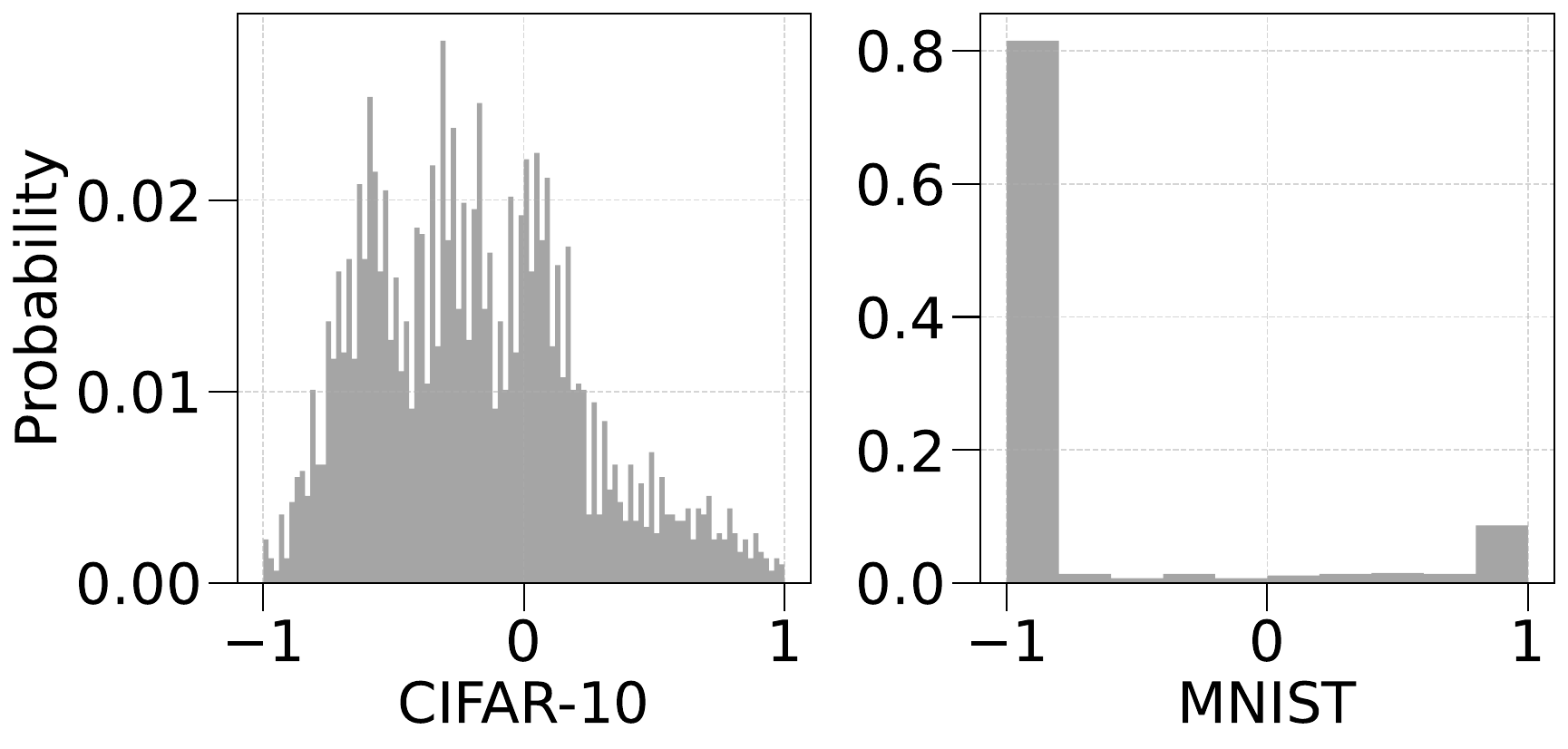}
  \caption{The distribution of the first image of CIFAR-10 and MNIST datasets.}
  \label{fig:exp2-data-distribution}
\end{figure}

We evaluated our metrics on seven real-world datasets, covering tabular, text, and image datasets, as shown in Table~\ref{tab:datasets}. All image data are first rescaled to [0,1] and then normalized to [-1,1] with a variance and mean of 0.5. For two tabular datasets, we one-hot encode the discrete columns and rescale the continuous columns to [-1,1]. For CelebA \cite{liu2015faceattributes}, we built the binary attractiveness classification task following \cite{nips2023gan} and resize the image to 112x112x3. For TinyImageNet, we built 4-class classification task follow \cite{sec-xu2025risk}.

\underline{CIFAR10} \cite{krizhevsky2009learning}: Comprising 50,000 training samples and 10,000 test samples across 10 classes, {\bf CIFAR-10} dataset is a benchmark in computer vision with categories like airplanes, automobiles, and birds. Each image is a 32x32 pixel color image with three channels.

\underline{MNIST} \cite{lecun1998gradient}: Comprising 60,000 training samples and 10,000 test samples across 10 digits, with each image being a 28x28 gray-scale image. Fig.~\ref{fig:exp2-data-distribution} plots the distribution of the first image of CIFAR-10 and MNIST, respectively. We bin the continuous pixel value of two images into 100 (for CIFAR-10) and 10 (for MNIST) discrete intervals between -1 and 1. We can see that the probability distribution of pixel values in MNIST is more concentrated around specific values, while CIFAR-10 shows a more uniform spread across the range.

\underline{CelebA} \cite{liu2015faceattributes}: is a large-scale facial attributes dataset containing over 200,000 celebrity images, which are divided into training set 162770 (proportion 0.8), verification set 19867 (proportion 0.098), and test set 19962 (proportion 0.096). Each image has 40 binary attribute annotations and 5 landmark locations. It is widely used as a benchmark for tasks such as face recognition, facial attribute prediction, and generative adversarial networks (GANs) due to its large diversity, rich annotations, and high quality.

\underline{TinyImageNet} \cite{le2015tiny}: is a downsized and more computationally manageable variant of the ImageNet Large Scale Visual Recognition Challenge (ILSVRC). It contains 100,000 images across 200 object classes, with each class having 500 training, 50 validation, and 50 test images.

\underline{Purchase100} (Purchase) \cite{sec21-evaluation-ml}: is based on Kaggle’s Acquire Valued Shoppers Challenges, which has 600 binary features indicating whether a good has been purchased or not, for identifying different shopping styles of a customer. We use the simplified version proposed by \cite{sec21-evaluation-ml}.

\underline{Home Credit} (Credit) \cite{home-credit-kaggle}: is based on Kaggle’s Home Credit Default Risk Challenge, to predict the credit applicant’s default likelihood.

\underline{GLUE/SST-2} \cite{wang2018glue}: GLUE is a multi-task benchmark for natural language understanding, where the SST-2 is designed for sentiment analysis.


\subsection{Models}
\begin{table*}[!thpb]
    \caption{Statistics of Models. Convolution layers are denoted by C, followed by the number of filters; MaxPooling, AveragePooling, and Adaptive AveragePooling layers are MP, AP, and AAP; Fully Connected layer is FC with the number of neurons; Basic Block of ResNet is BB with the channel size. EMB is the embedding layer of the transformer, and Tr is the transformer block. } 
    \label{tab:models}

    \centering
    \begin{tabular}{ccc}
    \toprule
    Model  & \# Critical Layer & Architecture \\
    \midrule
    \multirow{2}{*}{VGG-5 for CIFAR-10}& \multirow{2}{*}{7} & C32-MP-C64(default)-MP- \\
    & & C64-FC128-FC10\\
    \multirow{2}{*}{VGG-9 for CIFAR-10}& \multirow{2}{*}{14}& C64-MP-C64-C128-MP(default)-C128-\\
    &&C256-MP-C256-MP-C256-FC4096-FC10\\
    \multirow{2}{*}{ResNet-18 for CIFAR-10}& \multirow{2}{*}{12}&C64-AP-[BB64]*2-[BB128]*2 (default)- \\
    &&[BB256]*2-[BB512]*2-AAP-FC10\\
    \multirow{2}{*}{ResNet-34 for CIFAR-10}& \multirow{2}{*}{20}& C64-AP-[BB64]*3-[BB128]*4 (default)-\\
    &&[BB256]*6-[BB512]*3-AAP-FC10\\
    FCN for Purchase& 5 & FC1024-FC512(default)-FC256-FC128-FC100\\
    FCN for Credit & 4 & FC512-FC128(default)-FC32-FC1\\
    DistilBert for GLUE/SST-2& 8 & EMB768-Tr768(defulat)-[Tr768]*4-FC768-FC2 \\
    \bottomrule
    \end{tabular}

\end{table*}
We use seven models, including VGG-5 \cite{simonyan2014very}, VGG-9 \cite{simonyan2014very}, ResNet-18 \cite{he2016deep}, ResNet-34 \cite{he2016deep}, a 5-layer linear model for Purchase, a 4-layer linear model for Credit, and DistilBert \cite{sanh2019distilbert}, as shown in Table.~\ref{tab:models}. For image data, we only give the model architecture for CIFAR-10 as a representative case, since the architectures for other image datasets involve only minor modifications on the input and output layers.

The number of critical layers and architecture of these models are presented in Table~\ref{tab:models}. We also identify the default split layer in the brackets. All models except for DistilBert are trained until the test accuracy tends to stabilize by default (approximately 20 epochs). For DistilBert, we use the pretrained weights from HuggingFace.

\section*{Detailed Implementations for FSInfo calculation}
We give the detailed implementation of FSInfo in Table~\ref{tab:implementation}.
\begin{table*}[h]
    \caption{Numerical calculation procedure for FSInfo.} 
    \label{tab:implementation}

    \centering
    \begin{tabular}{c p{14cm}}
    \toprule
     Step & Actions \\
    \midrule
1 &
Calculate the $d_x log(2\pi e)$.\\
2  &
Calculate the diagonal Fisher information Matrix (FIM).\\
2.1  &
\textbf{For the small models (VGG-5, VGG-9, ResNet-18, ResNet-34, and MLPs for Credit and Purchase), we use intuitive and simple implementations:}\\
2.1.1 &
Calculate the Jacobian of $z$ w.r.t. $x$, $J_{f_{\theta 1}(x)}$ , using "torch.autograd.functional.jacobian".\\
2.1.2 &
Calculate the Fisher Information Matrix $F_{z|x} = \frac{1}{\sigma^2} J_{f_{\theta 1}(x)}^\top J_{f_{\theta 1}(x)}$ and get the diagonal vectors $\Lambda=(\lambda_1,\dots,\lambda_{d_x})$.\\
2.2  &
\textbf{For big models (DistilBert), we want to save the memory cost:}\\
2.2.1 &
Only calculate the diagonal item of $J_{f_{\theta 1}(x)}^\top J_{f_{\theta 1}(x)}$ instead of calculating the whole matrix by using "torch.autograd.functional.jvp" for memory efficiency. Then we get the diagonal vectors of  $J_{f_{\theta 1}(x)}^\top J_{f_{\theta 1}(x)}$.\\
2.2.2 &
Multiply each diagonal item of $J_{f_{\theta 1}(x)}^\top J_{f_{\theta 1}(x)}$ by $\frac{1}{\sigma^2}$. Then we get the diagonal vectors $\Lambda=(\lambda_1,\dots,\lambda_{d_x})$.\\
3 &
Calculate the term $\log\det(F_{z|x} )=\Sigma_i^{d_x} \log (\lambda_i + 1*e^{-10)})$.\\
4 &
Get the FSInfo for $x$ FSInfo = $-\frac{1}{2d_x}[d_x log(2\pi e)-\log\det(F_{z|x} )]$.\\
    \bottomrule
    \end{tabular}

\end{table*}

\section*{Proof of Theorem 4.2}\label{app:proof}
Tan et al. \cite{sec24-defending-FL-ita} have modeled the estimation error of the DRA adversaries who aim at reconstructing the training dataset $D\in\mathbb{R}^d$ of one client from the model parameters $W\in\mathbb{R}^m$ in the federated learning scenarios by giving Theorem~1 in \cite{sec24-defending-FL-ita}:
\begin{theorem} \label{theorem:Tan-FL}
({Theorem~1 in \cite{sec24-defending-FL-ita}}). For any random variable $D\in\mathbb{R}^d$ and $W\in\mathbb{R}^m$, we have
\begin{equation}
\mathbb{E}[\Vert D-\hat{D}(W)\Vert^2/d] \geq \frac{e^{2h(D)/d}}{2\pi e}e^{-2I(D;W)/d},
\end{equation}
where $\hat{D}(W)$ is an estimator of $D$ constructed by $W$.
\end{theorem}
By simply applying Theorem~\ref{theorem:Tan-FL} in the privacy quantification in smashed data about raw input in split inference settings, we reach Theorem~\ref{theorem:REB-Tan}. For better understanding, we follow \cite{sec24-defending-FL-ita} to give the proof details of Theorem~\ref{theorem:REB-Tan} as follows:

To prove Theorem~4.2, we first introduce the following lemma:
\begin{lemma}\label{lemma:1}
    {(Lemma 2 in \cite{sec24-defending-FL-ita}).} For any $d_x$-dimensional semi-positive definite matrix $A$ ($A \in \mathbb{R}^{d_x\times d_x}$), we have:
    \begin{equation}
        det(A) \leq (\frac{tr(A)}{d_x})^{d_x},
    \end{equation}
    where the $det(\cdot)$ is the determinant and the $tr(\cdot)$ denotes the trace.
\end{lemma}

Now, we prove the Theorem~4.2 as below:
\begin{proof}
Let $X$ be the raw input, $\hat{X}$ be the reconstructed input, and $x\sim p(X)$ we have:
\begin{align} 
    & H(\hat{X}|X) \\
     =& \mathbb{E}_{x} H(\hat{X}|X=x) \\
    \leq & \mathbb{E}_{x} \left[ \frac{d_x}{2} \log (2\pi e) + \frac{1}{2} \log \det (Cov(\hat{X}|X=x)) \right] \label{eq:maxentropy} \\
    =& \frac{d_x}{2} \log (2\pi e) + \mathbb{E}_{x} \left[ \frac{1}{2} \log \det (Cov(\hat{X}|X=x)) \right] \\
    \leq & \frac{d_x}{2} \log (2\pi e) + \mathbb{E}_{x} \left[ \frac{d_x}{2} \log \frac{tr(Cov(\hat{X}|X=x))}{d_x} \right] \label{eq:lemma1}\\
    = & \frac{d_x}{2} \log (2\pi e) + \mathbb{E}_{x} \left[ \frac{d_x}{2} \log \frac{\mathbb{E}_{p(X,\hat{X})}\left[ ||\hat{X}-X||^2 |X=x \right]}{d_x} \right] \label{eq:unbiased} \\    
    \leq & \frac{d_x}{2} \log (2\pi e) +  \frac{d_x}{2} \log \frac{\mathbb{E}_{x} \left[ \mathbb{E}_{p(X,\hat{X})}\left[ ||\hat{X}-X||^2 |X=x \right] \right]}{d_x} \label{eq:Jensen}\\
    =& \frac{d_x}{2} \log (2\pi e) +  \frac{d_x}{2} \log \frac{\mathbb{E}_{p(X,\hat{X})}\left[ ||\hat{X}-X||^2 \right]}{d_x}.
\end{align}
The first inequality holds because the Gaussian distribution is the maximum entropy distribution among distributions with the same mean value and covariance. The second inequality depends on Lemma~\ref{lemma:1}. The Equation~\ref{eq:unbiased} holds because,
\begin{equation}
    tr(Cov(\hat{X} \vert X=x))=\sum_{i=1}^{d_x}\mathbb{E}_{p(\hat{X_i})}[(\hat{X_i}-\mathbb{E}[\hat{X_i}])^2 \vert X=x],
\end{equation}
where the mean vector $\mathbb{E}_{p(\hat{X_i})} = X$ when the DRA adversary is unbiased, and the second-order norm of $(\hat{X}-X)$ is $\sqrt{\sum_{i=1}^{d_x}(\hat{X}_i-X_i)^2}$. The third inequality stems from Jensen's inequality, where for all concave functions $f$ (e.g., $\log$), $f(\mathbb{E}[x]) \geq \mathbb{E}[f(x)]$.
Hence, we can have:
\begin{equation}
    \frac{\mathbb{E}_{p(X,\hat{X})}[\Vert X-\hat{X} \Vert^2]}{d_x} \geq \frac{e^{\frac{2}{d_x}(H(\hat{X}\vert X))}}{2\pi e}.
\end{equation}

\end{proof}

\section*{Proof of Theorem 4.4}\label{app:proof-REB-minimax}



First, we have the following Lemma:
\begin{lemma}\label{lemma:3} (Theorem 9 in \cite{scarlett2019introductory}). Under the minimax estimation setup in Section~4.2, fix $\epsilon>0$, and index $V\in\mathcal{V}$ is drawn from a prior distribution $P_V$. $X_V=\{x_{v_1},\dots,x_{v_{|\mathcal{V}|}}\}$ drawn from a probability distribution $P_{X|V}$ be a finite subset of $\mathcal{X}$ such that
\begin{equation}
\rho(x_v,x_{v'})\geq \epsilon, \forall v,v' \in \mathcal{V}, v\neq v'.
\end{equation}
Then, we have
\begin{equation}
\mathcal{M}_n(\mathcal{X},\ell) \geq \Phi(\frac{\epsilon}{2})(1-\frac{I(V;\hat{V})+\log 2}{log |\mathcal{V}|}),
\end{equation}
where $V$ is uniform on $\mathcal{V}$, and the mutual information is with respect to $V \rightarrow X_V \rightarrow Z \rightarrow \hat{X} \rightarrow \hat{V}$.
\end{lemma}
\begin{proof}
In Lemma~\ref{lemma:3}, the $X_V$ is identified by $V$, so we let both $V$ and $X_V$ be the raw input $X$ in split inference. In this case, given the $V=X$, we have $X_V=V=X$, where the $X_V$ is uniquely identified by $V$. 
Then, we let $Z$ be the smashed data $Z$ and both $\hat{X}$ and $\hat{V}$ be the reconstructed input $\hat{X}$ in split inference system, and we have the Markov chain $$X \rightarrow X \rightarrow Z \rightarrow \hat{X} \rightarrow \hat{X}. $$ Follow \cite{noorbakhsh2024inf2guard}, we assume $X$ is uniform over $\mathcal{X}$. Then by applying Lemma~\ref{lemma:3} and using $I(X;\hat{X})=H(\hat{X})-H(\hat{X}|X)$, we reach out Theorem~4.4.
\end{proof}

\bibliographystyle{IEEEtran}
\bibliography{main}
%

 
\vspace{11pt}

\begin{IEEEbiography}[{\includegraphics[width=1in,height=1.25in,clip,keepaspectratio]{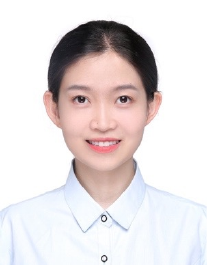}}]{Ruijun Deng}
received her B.S. degree in Computer Science from Fudan University in 2022. Currently, she is pursuing the Master of Computer Science degree with Fudan University, China. Her research interests include federated learning, split learning, and AI security.
\end{IEEEbiography}
\vspace{10pt}

\begin{IEEEbiography}[{\includegraphics[width=1in,height=1.25in,clip,keepaspectratio]{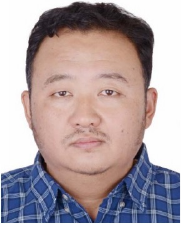}}]{Zhihui Lu}
(Member, IEEE) is a Professor at the School of Computer Science, Fudan University. He received a Ph.D. computer science degree from Fudan University in 2004 and is a member of the IEEE and China Computer Federation’s service computing specialized committee. His research interests are cloud computing and service computing technology, big data architecture, edge computing, and IoT distributed systems. He has (co-)authored two books and more than 100 journal articles and conference papers in these areas.
\end{IEEEbiography}
\vspace{10pt}

\begin{IEEEbiography}[{\includegraphics[width=1in,height=1.25in,clip,keepaspectratio]{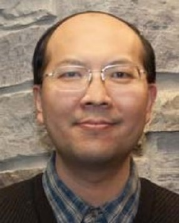}}]{Qiang Duan}
(Senior Member, IEEE) is currently a Professor with the College of Information Sciences and Technology, The Pennsylvania State University, Abington, PA, USA. His current research interests include network virtualization and softwarization, cognitive and autonomous networking, and edge computing-based ubiquitous intelligence. He has published three books and more than 100 research papers in these areas. He has served as an editor/associate editor for multiple research journals.
\end{IEEEbiography}

\begin{IEEEbiography}[{\includegraphics[width=1in,height=1.25in,clip,keepaspectratio]{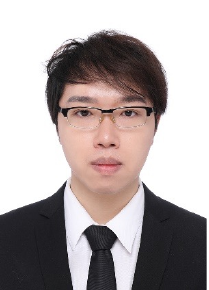}}]{Shijing Hu}
received his B.S. in Computer Science from Peking University in 2021. Currently, he is pursuing doctor's degree at the School of Computer Science, Fudan University, China. His research interests include edge intelligence and large language models. 
\end{IEEEbiography}
\vspace{10pt}




\vfill

\end{document}